\documentclass[12pt, onecolumn]{IEEEtran}
\usepackage{epsfig,latexsym}
\usepackage{float}
\usepackage{indentfirst}
\usepackage{amsmath}
\usepackage{amssymb}
\usepackage{times}
\usepackage{subfigure}
\usepackage{psfrag}
\usepackage{cite}
\usepackage{lastpage}
\usepackage{fancyhdr}
\usepackage{color}
 \usepackage{amsthm}
\usepackage{bigints}
\newcommand{\block}[1]{
  \underbrace{\begin{matrix}0 & \cdots & 0\end{matrix}}_{#1}
}
\newtheorem{theorem}{Theorem}
\newtheorem{Lemma}{Lemma}
\newtheorem{lemma}[Lemma]{$\mathbf{Lemma}$}

\begin{document}
\title{ {\huge  Design of Massive-MIMO-NOMA with Limited Feedback }}

\author{ Zhiguo Ding, \IEEEmembership{Senior Member, IEEE} and  H. Vincent Poor, \IEEEmembership{Fellow, IEEE}\thanks{
The authors   are with the Department of
Electrical Engineering, Princeton University, Princeton, NJ 08544,
USA.   Z. Ding is also with the School of
Computing and Communications, Lancaster
University, LA1 4WA, UK.
}\vspace{-2em}} \maketitle
\begin{abstract}
In this letter, a low-feedback non-orthogonal multiple access (NOMA) scheme using massive multiple-input multiple-output (MIMO) transmission is proposed. In particular, the proposed scheme can decompose a massive-MIMO-NOMA system  into multiple separated single-input single-output NOMA channels, and analytical results are developed to evaluate the performance of the proposed scheme for   two scenarios, with perfect user ordering and with one-bit feedback, respectively.
\end{abstract}\vspace{-1.2em}
\section{Introduction}
Non-orthogonal multiple access (NOMA) has recently been recognized as a key multiple access (MA) technique for solving the spectrum crunch faced by   next generation mobile networks \cite{NOMAPIMRC,Nomading,7015589}. Unlike conventional orthogonal MA (OMA), NOMA uses the power  domain to serve users in the same time/code/frequency channels,  where successive interference cancellation (SIC) is employed at  users with better channel conditions.  The impact of dynamic user locations on the performance of NOMA was investigated in \cite{Nomading}, and realizing different user fairness criteria in NOMA systems has been studied in \cite{Krikidisnoma}.  The use of multiple-input multiple-output (MIMO) technologies can further improve the performance of NOMA, and the optimal design of MIMO-NOMA has been investigated in  \cite{7015589} and \cite{7095538}.

Most of existing work on MIMO-NOMA, such as \cite{7015589} and \cite{7095538}, has assumed  perfect knowledge of channel state information (CSI) at the transmitter, which is difficult to realize in practice. Particularly when massive MIMO is applied to NOMA, the perfect CSI assumption can consume excessive bandwidth resources \cite{6736761}. The scheme proposed in  \cite{Zhiguo_mimoconoma} does not need CSI at the transmitter,  but requires that the number of the receive antennas is larger that of the transmitter.  In this paper   a massive-MIMO-NOMA downlink transmission protocol is proposed, which does not require the users to feed their channel matrices back to the base station. In particular, by employing the spatial clustering structure of users' channel matrices and applying the concept of    MIMO-NOMA in \cite{Zhiguo_mimoconoma}, the massive-MIMO-NOMA system can be decomposed into separated single-input single-output (SISO) NOMA channels. Then   the    performance analysis  for the proposed protocol is carried out for two scenarios, with perfect user ordering and with one-bit feedback, respectively. Exact expressions for the outage probabilities for the two scenarios are developed and the achievable diversity order is also obtained by applying high signal-to-noise ratio (SNR) approximations.

\section{Description of Massive-MIMO-NOMA}
Consider a downlink transmission scenario with one base station communicating with multiple users. The base station is equipped with $M$ antennas, and each user has $N$ antennas, where $M>>N$.  Following the geometrical one-ring scattering model used in \cite{6542746} and \cite{Dai15}, we divide the users into $K$ spatial clusters, in which  there are $L$ users in each cluster sharing the same spatial correlation matrix, denoted by $\mathbf{R}_{k}$. Therefore for the $l$-th user in the $k$-th cluster,  its channel matrix, denoted by $\mathbf{H}_{k,l}$,  can be decomposed  as  $\mathbf{H}_{k,l} = \mathbf{G}_{k,l}\mathbf{\Lambda}_k^{\frac{1}{2}}\mathbf{U}_k$, where $ \mathbf{G}_{k,l}$ denotes an $N\times r_k$ fast-fading complex Gaussian matrix,   $\mathbf{\Lambda}_k$ is a $r_k \times r_k$ diagonal matrix containing $r_k$ non-zero eigenvalues of $\mathbf{R}_k$, and  $\mathbf{U}_k$, a  $r_k\times M$ matrix,  collects all the corresponding eigenvectors of $\mathbf{R}_{k}$, i.e.,  $ \mathcal{E}\{\mathbf{H}_{k,l}^H\mathbf{H}_{k,l}\}= \mathbf{U}_k^H\mathbf{\Lambda}_k   \mathbf{U}_k = \mathbf{R}_{k}$.

While the CSI about fast varying $\mathbf{G}_{k,l}$ is difficult to obtain at the transmitter, $\mathbf{R}_k$ represents channel correlation and varies slowly, so it is more reasonable to assume that the base station has access to $\mathbf{R}_k$.
To use this spatial clustering  feature, the base station will send the following superimposed signal:
\begin{align}
\mathbf{s} = \sum^{K}_{k=1} \mathbf{P}_k \sum^{L}_{l=1}\mathbf{w}_{k,l}s_{k,l},
\end{align}
where the $M\times \tilde{M}_k$ precoding matrix, $\mathbf{P}_k$, is used to avoid the inter-cluster interference, $s_{k,l}$ is the message to the $l$-th user in the $k$-th cluster, and the $\tilde{M}_k\times 1$   vector, $\mathbf{w}_{k,l}$, is the precoding vector for the $l$-th user in the $k$-th cluser. For notational  simplicity, we assume  that $\tilde{M}_k$ and $r_k$ are the same among all clusters, and are denoted by $\tilde{M}$ and $r$, respectively.
\vspace{-1.5em}
\subsection{Inter-Cluster Interference Cancellation}
The $l$-th user in the $k$-th cluster observes the following
\begin{align}\label{original model}
\mathbf{y}_{k,l} =\mathbf{G}_k\mathbf{\Lambda}_k^{\frac{1}{2}}\mathbf{U}_k  \sum^{K}_{i=1} \mathbf{P}_i \sum^{L}_{j=1}\mathbf{w}_{i,j}s_{i,j} +\mathbf{n}_{k,l},
\end{align}
where $\mathbf{n}_{k,l}$ is the noise vector. In order to avoid   inter-cluster interference, $\mathbf{P}_k$ needs to satisfy the following constraint:
\begin{align}\label{xx2}
\begin{bmatrix} \mathbf{U}_{1}^H &\cdots &\mathbf{U}_{k-1}^H  &\mathbf{U}_{k+1}^H &\cdots &\mathbf{U}_{K}^H \end{bmatrix}^H\mathbf{P}_k =\mathbf{0},
\end{align}
where $\tilde{M}= (M-r(K-1))$.
By using   $\mathbf{P}_k$ in \eqref{xx2}, the system model can be simplified as follows:
  \begin{align}\label{small mimo}
\mathbf{y}_{k,l}& =\underset{N\times \tilde{M}}{\underbrace{\mathbf{G}_{k,l}\mathbf{\Lambda}_k^{\frac{1}{2}}\mathbf{U}_k   \mathbf{P}_k }} \sum^{L}_{j=1}\mathbf{w}_{k,j}s_{k,j} +\mathbf{n}_{k,l},
\end{align}
where the inter-cluster interference is removed.
\vspace{-1em}
\subsection{The Application of the NOMA Approach}
As can be observed from \eqref{small mimo}, the system model with  massive MIMO has been decomposed into $K$ separated ones with $\tilde{M}$ effective transmit antennas for each cluster. There are still two difficulties in the simplified   NOMA  system. One is how to accommodate   $L$ users in one cluster with $\tilde{M}$ transmit antennas, particularly if $L>\tilde{M}$. The other is that the base station has no access to the instantaneous realization of $\mathbf{G}_{k,l}$.

The MIMO-NOMA concept proposed in \cite{Zhiguo_mimoconoma} can be applied here. Particularly  we can further divide the $L$ users in one cluster into $Q$ groups of size $P$, i.e. $L=PQ$. Denote the signal transmitted to the $L$ users in the $k$-th cluster by $\tilde{\mathbf{s}}_k\triangleq \sum^{L}_{l=1}\mathbf{w}_{k,l}s_{k,l}$ which   can be further written as follows:
  \begin{align}
\tilde{\mathbf{s}}_k=\sum^{Q}_{q=1}\sum^{P}_{p=1}\mathbf{w}_{k,q,p} s_{k,q,p},
\end{align}
where $s_{k,q,p}$ is a message to the $p$-th user in the $q$-th group in the $k$-th cluster and $\mathbf{w}_{k,q,p}$ denotes the corresponding precoding vector.

For the users within the same group, NOMA is applied, i.e. the users share the same precoding vector but with different power allocation coefficients, i.e.
  \begin{align}
\tilde{\mathbf{s}}_k= \sum^{Q}_{q=1}\mathbf{w}_{k,q}\sum^{P}_{p=1} \alpha_{k,q,p} s_{k,q,p},
\end{align}
where, for easiness of notation, we assume $\sum^{P}_{p=1} \alpha_{k,q,p}^2=1$  and $|\mathbf{w}_{k,q}|^2=1$, i.e., the total transmission power at the base station is $QK$. Without knowing the CSI at the transmitter, we simply let
$ \mathbf{w}_{k,q}=
  \begin{bmatrix}
  \smash[b]{\block{q-1}} & 1 & \smash[b]{\block{\tilde{M}-q}}
  \end{bmatrix}^T
$.
\vspace{-1em}
\subsection{Receiver Detection}
With the aforementioned design, the observation at the   $p$-th user in the $q$-th group in the $k$-th cluster can be expressed as follows:
  \begin{align}
\mathbf{y}_{k,q,p}& =\tilde{\mathbf{H}}_{k,q,p}\sum^{Q}_{m=1}\mathbf{w}_{k,m}\sum^{P}_{n=1} \alpha_{k,m,n} s_{k,m,n} +\mathbf{n}_{k,q,p},
\end{align}
where
$\tilde{\mathbf{H}}_{k,q,p}=\mathbf{G}_{k,q,p}\mathbf{\Lambda}_k^{\frac{1}{2}}\mathbf{U}_k   \mathbf{P}_k$. Assume $N\geq \tilde{M}$.  By applying the zero forcing approach, the user observes the following:
  \begin{align}\label{final system model}
\left[\tilde{\mathbf{H}}_{k,q,p}^{\dagger}\mathbf{y}_{k,q,p}\right]_{q,1}& = \sum^{P}_{n=1} \alpha_{k,q,n} s_{k,q,n} +\tilde{\mathbf{H}}_{k,q,p}^{\dagger}\mathbf{n}_{k,l},
\end{align}
where $\tilde{\mathbf{H}}_{k,q,p}^{\dagger}=(\tilde{\mathbf{H}}_{k,q,p}^H\tilde{\mathbf{H}}_{k,q,p})^{-1}\tilde{\mathbf{H}}_{k,q,p}^H$ and $[\mathbf{A}_{i,j}]$ denotes the element at the $i$-th row and $j$-th column of $\mathbf{A}$.
The covariance of the noise vector is given by
  \begin{align}
&\mathbf{C}_{k,q,p}= \tilde{\mathbf{H}}_{k,q,p}^{\dagger} (\tilde{\mathbf{H}}_{k,q,p}^{\dagger} )^H
\\\nonumber &=  ( \mathbf{P}_k^H\mathbf{U}_k^H\mathbf{\Lambda}_k^{\frac{1}{2}} \mathbf{G}_{k,q,p}^H  \mathbf{G}_{k,q,p}\mathbf{\Lambda}_k^{\frac{1}{2}}\mathbf{U}_k   \mathbf{P}_k)^{-1}=  ( \mathbf{P}_k^H \mathbf{R}_k   \mathbf{P}_k)^{-1}.
\end{align}

Comparing \eqref{final system model} to \eqref{original model}, the massive-MIMO-NOMA system model has been decomposed into separated SISO-NOMA channels, without knowing  $\mathbf{G}_{k,q,p}$ at the base station.

\section{Performance Analysis}
\subsection{With Perfect Knowledge of User Ordering}
Suppose that the base station knows   that    the users  are ordered  as follows:
\begin{align}\label{order}
\frac{1}{[\mathbf{C}_{k,q,1}]_{q,q}} \leq \cdots \leq \frac{1}{[\mathbf{C}_{k,q,P}]_{q,q}}.
\end{align}

Therefore  SIC can be carried out at the users. Particularly within the $q$-th group of the $k$-th cluster, the message to the $n$-th user can be decoded at the $p$-th user with the following signal-to-interference-plus-noise (SINR):
  \begin{align}
SINR_{k,q,p}^n&= \frac{\rho \frac{1}{[\mathbf{C}_{k,q,p}]_{q,q}}  \alpha_{k,q,n}^2  }{1+\sum^{P}_{m=n+1}\rho \frac{1}{[\mathbf{C}_{k,q,p}]_{q,q}}  \alpha_{k,q,m}^2 },
\end{align}
for $1\leq n\leq p\leq P$ except  $n=p=P$. When  $n=p=P$, the SINR becomes the following SNR:
\begin{align}
SINR_{k,q,P}^P&=  \rho \frac{1}{[\mathbf{C}_{k,q,P}]_{q,q}}  \alpha_{k,q,P}^2.
\end{align}
Denote the $p$-th user in the $q$-th group of the $k$-th cluster by $U_{k,q,p}$.  The following theorem provides the outage performance achieved by the proposed massive-MIMO-NOMA scheme.
\begin{theorem}\label{theorem1}
The outage probability at $U_{k,q,p}$ achieved by the proposed massive-MIMO-NOMA scheme is given by
{\small \begin{align}
\mathrm{P}_{k,q,p} &=   \sum^{p-1}_{i=0}{p-1 \choose i} (-1)^i  \pi_p^P \frac{1-\left(F(\rho\xi^*_{k,q,p}) \right)^{P-p+i+1}}{P-p+i+1} ,
\end{align}}
$\hspace{-0.8em}$if $\xi^*_{k,q,p}\geq 0$, otherwise $\mathrm{P}_{k,q,p}^n=1$,   where $\xi^*_{k,q,p}=\min \left\{ \xi_{k,q,n}, 1\leq n \leq p \right\}$,  $a_{k,q} = \frac{1}{[( \mathbf{P}_k^H\mathbf{R}_k \mathbf{P}_k )^{-1}]_{q,q}}$, $\pi_p^P= \frac{P!}{(P-p)!(p-1)!}$, $R_{k,q,n}$ is the targeted rate for $U_{k,q,n}$,  $\tau_{k,q,n}=2^{R_{k,q,n}}-1$,   $\gamma(\cdot)$ denotes the incomplete gamma function,
\begin{align}\xi_{k,q,n}=\left\{\begin{array}{ll}\frac{\alpha_{k,q,P}^2   }{\tau_{k,q,P}}, &  \text{if}~ $n=P$\\\frac{ \alpha_{k,q,n}^2 - \tau_{k,q,n} \sum^{P}_{m=n+1}  \alpha_{k,q,m}^2 }{\tau_{k,q,n}}, & \text{otherwise}\end{array} \right.,
\end{align}
and \begin{align}
F(\rho\xi_{k,q,n})  =   1- \frac{\gamma\left(N-\tilde{M}+1, \frac{1}{\rho a_{k,q}\xi_{k,q,n}}\right)}{\Gamma(N-\tilde{M}+1)}.
\end{align}
In addition, the diversity order achieved   at $U_{k,q,p}$ is $p(N-\tilde{M}+1)$.
\end{theorem}
\begin{proof} The proof can be divided into three steps as follows:
\subsubsection{Characterizing $SINR_{k,q,p}^n$}
To evaluate the outage probability, it is important to find the density function of $SINR_{k,q,p}^n$, where the key step is characterize $[\mathbf{C}_{k,q,p}]_{q,q}$.

First suppose users within the same group are {\it not ordered}. Since $\mathbf{G}_{k,q,p}$ are complex Gaussian distributed,
$\mathbf{G}_{k,q,p}^H  \mathbf{G}_{k,q,p}$ follows the complex Wishart distribution, i.e., $\mathbf{G}_{k,q,p}^H  \mathbf{G}_{k,q,p}\sim \mathcal{W}^C_{r}\left(N, \mathbf{I}_{r}\right)$ \cite{861781}. Therefore, the matrix in the covariance matrix, $\mathbf{C}_{k,q,p}$, is also complex Wishart distributed, i.e.,
\begin{align}\nonumber
 \mathbf{P}_k^H\mathbf{R}_k   \mathbf{P}_k \sim \mathcal{W}^C_{\tilde{M}}\left(N, \mathbf{P}_k^H\mathbf{R}_k \mathbf{P}_k \right),
\end{align}
which means that $\mathbf{C}_{k,q,p}$ is inverse Wishart distributed, i.e.,
\begin{align}
 \mathbf{C}_{k,q,p} \sim \mathcal{W}^{C^{-1}}_{\tilde{M}}\left(N, \mathbf{P}_k^H\mathbf{R}_k \mathbf{P}_k \right).
\end{align}
By using the marginal distribution for the  elements on the diagonal of an inverse Wishart matrix, the probability density function (pdf) of $[ \mathbf{C}_{k,q,p}]_{q,q}$ is given by \cite{861781}
\begin{align}
f (x) = \frac{x^{-(N-\tilde{M}+2)}}{\Gamma(N-\tilde{M}+1)a_{k,q}^{N-\tilde{M}+1}}e^{-\frac{1}{a_{k,q}x}},
\end{align}
where $a_{k,q} = \frac{1}{[( \mathbf{P}_k^H\mathbf{R}_k \mathbf{P}_k )^{-1}]_{q,q}}$. The cumulative distribution function (CDF), $F(x)$, can be found by integrating $f(x)$.

Because the $P$ users in each group are ordered as in \eqref{order}, the density function of $[ \mathbf{C}_{k,q,p}]_{q,q}$ at the $p$-th {\it ordered} user is given by \cite{David03}
\begin{align}
f_{p}(x) = \pi_p^Pf(x) \left(F(x) \right)^{P-p} \left(1-F (x) \right)^{p-1}.
\end{align}
where $\pi_p^P= \frac{P!}{(P-p)!(p-1)!}$.

\subsubsection{Calculating the outage probability}
The overall outage probability at $U_{k,q,p}$ is the event that this user cannot decode   the message   $s_{k,q,n}$, $1\leq n \leq p$.  Therefore the overall outage probability  at this user is given by
\begin{align}\nonumber
\mathrm{P}_{k,q,p} &= \mathrm{P}\left( \log(1+SINR_{k,q,p}^n)<R_{k,q,n}, \forall n\in \{1, \cdots, p\} \right)\\
 &=  \mathrm{P}\left([ \mathbf{C}_{k,q,p}]_{q,q}>\rho \xi^*_{k,q,p}\right),
\end{align}
where $\xi^*_{k,q,p}=\min \left\{ \xi_{k,q,n}, 1\leq n \leq p \right\}$.  By applying the pdf of $[ \mathbf{C}_{k,q,p}]_{q,q}$, the outage probability  is obtained as follows:
\begin{align}\nonumber
\mathrm{P}_{k,q,p} & =   \int_{\rho\xi^*_{k,q,p} }^{\infty} \pi_p^Pf(x) \left(F(x) \right)^{P-p} \left(1-F (x) \right)^{p-1}dx
\\ \nonumber & =  \sum^{p-1}_{i=0}{p-1 \choose i} (-1)^i \int_{\rho\xi^*_{k,q,p} }^{\infty} \pi_p^Pf(x) \left(F(x) \right)^{P-p+i} dx.
\end{align}
After some algebraic manipulations, an exact expression for the outage probability can be obtained as in the theorem.

\subsubsection{Obtaining the diversity order}
 At high SNR, i.e., $\rho$ approaches infinity,   $\frac{1}{\rho a_{k,q}\xi_{k,q,n}}$ approaches zero. Therefore, for a fixed $R_{k,q,n}$, we   have the following approximation:
\begin{align}\nonumber
F(\rho\xi_{k,q,n})  &=     e^{-\frac{1}{\rho a_{k,q}\xi_{k,q,n}}}\sum^{N-\tilde{M}}_{j=1}\frac{1}{j!\rho^j a_{k,q}^j\xi_{k,q,n}^j}\\ \nonumber &=   1- e^{-\frac{1}{\rho a_{k,q}\xi_{k,q,n}}} \sum^{\infty}_{j=N-\tilde{M}+1}\frac{1}{j!\rho^j a_{k,q}^j\xi_{k,q,n}^j} \\ \label{x1}  &\approx 1-\frac{1}{(N-\tilde{M}+1)!(\rho a_{k,q}\xi_{k,q,n})^{N-\tilde{M}+1}}.
\end{align}
By using the above approximation, the outage probability can be approximated at high SNR as follows:
 \begin{align}\nonumber
\mathrm{P}_{k,q,p} & =   \int_{\rho\xi^*_{k,q,p} }^{\infty} \pi_p^Pf(x) \left(F(x) \right)^{P-p} \left(1-F (x) \right)^{p-1}dx \\ \nonumber & \underset{(a)}{\rightarrow}   \int_{\rho\xi^*_{k,q,p} }^{\infty} \pi_p^Pf(x) \left(1-F (x) \right)^{p-1}dx  \\   & =  \pi_p^P  \left(1-F (\rho\xi^*_{k,q,p}) \right)^{p}\underset{(b)}\rightarrow \frac{1}{\rho^{p(N-\tilde{M}+1)}},
\end{align}
where both approximations, $(a)$ and $(b)$, are obtained due to \eqref{x1}. The proof is complete.
\end{proof}
\vspace{-1.5em}
\subsection{With One-Bit Feedback}
The scheme described  in the previous section does not need   global CSI about $\mathbf{H}_{k,q,p}$ at the base station; however, the base station still needs to  know how to order  the scalar effective channel gains $\frac{1}{[\mathbf{C}_{k,q,p}]_{q,q}}$, as in \eqref{order}. In this sub-section, we consider a scenario in which only one bit feedback is allowed. Particularly, each user will compare its effective channel gain to  a predefined threshold $\tau$, and feed one bit back to inform the base station whether its channel gain is below or above this threshold.
With one bit feedback, the users in the $q$-th group of the $k$-th cluster are divided into two sub-groups, denoted by $\mathcal{S}_1$ and $\mathcal{S}_2$, respectively\footnote{For notational simplicity, the subscripts of $k$ and $q$ are dropped in this subsection. }. When there is more than one user in a sub-group, the base station randomly orders the users, with predefined targeted data rates and power allocation coefficients. Due to   space limitations, we focus only on a special case with two users in each group (i.e. $P=2$). The outage probability for the weak user is given by
\begin{align}\label{p1}
\mathrm{P}_{1} &= \mathrm{P}\left(U_1\in \mathcal{S}_1, E\right) + \mathrm{P}\left(U_1\in \mathcal{S}_2, E\right)  ,
\end{align}
where $E$ denotes the event that outage occurs.
The first factor in the above equation can be written as follows:
\begin{align}
\mathrm{P}\left(U_1\in \mathcal{S}_1, E\right)  &=     \mathrm{P}\left(U_1\in \mathcal{S}_1,|\mathcal{S}_1|=1, E\right) \\ \nonumber & +\sum^2_{i=1}\mathrm{P}\left(U_1\in \mathcal{S}_1,|\mathcal{S}_1|=2 , E_i\right),
\end{align}
where $E_i$ denotes the event that   the user is randomly put as the $i$-th NOMA user and outage occurs.
The above probability can be calculated as follows:
  \begin{align}
\mathrm{P}\left(U_1\in \mathcal{S}_1, E\right)  &=     \mathrm{P}\left( [\mathbf{C}_{1}]_{q,q}> \phi_1, [\mathbf{C}_{2}]_{q,q}<\tilde{\tau}\right) \\ \nonumber & +\frac{1}{2}\sum^2_{i=1} \mathrm{P}\left( [\mathbf{C}_{1}]_{q,q}>\phi_i, [\mathbf{C}_{2}]_{q,q}>\tilde{\tau}\right),
\end{align}
where $\tilde{\tau}=\frac{1}{\tau}$ and $\phi_i=\max\left\{\tilde{\tau},\rho\xi^*_{i}\right\}$.
The joint pdf of the two effective channels is given by $f_{[\mathbf{C}_{1}]_{q,q},[\mathbf{C}_{2}]_{q,q}}(x,y)=2f(x)f(y)$, which means that the above probability can be found as follows:
\begin{align}\nonumber
&\mathrm{P}\left(U_1\in \mathcal{S}_1, E\right)  =   2 \left(1- F\left( \phi_1\right)\right)F\left(\tilde{\tau}\right)\\ \nonumber &+\frac{1}{2}\sum^2_{i=1} \left(1 - \left[F\left( \phi_i\right)\right]^2\right)  - \sum^2_{i=1} F\left(\tilde{\tau}\right)\left(1 - F\left(\phi_i\right)\right),
\end{align}
for $\xi^*_{i}\geq0$.
The second factor in \eqref{p1} can be written as follows:
\begin{align}
\mathrm{P}\left(U_1\in \mathcal{S}_2, E\right)  &= \frac{1}{2} \sum^2_{i=1} \mathrm{P}\left( \tilde{\tau}> [\mathbf{C}_{1}]_{q,q}>\rho\xi^*_{i}  \right),
\end{align}
for  $\tilde{\tau}>  \rho\xi^*_{i}$, for $i\in\{1, 2\}$.
The marginal  CDF of the weak channel gain (strong noise) is given by $F_{[\mathbf{C}_{1}]_{q,q}}(x)=(F(x))^2$, which means the above probability can be found as follows:
\begin{align}
\mathrm{P}\left(U_1\in \mathcal{S}_2, E\right)  &= \left[F\left( \tilde{\tau}\right)\right]^2 -  \frac{1}{2} \sum^2_{i=1}  \left[F\left(\rho\xi^*_{i}  \right)\right]^2,
\end{align}
for  $\tilde{\tau}>  \rho\xi^*_{i}$, for $i\in\{1, 2\}$.
Therefore the overall outage probability for this  weak user is given by
\begin{align}
\mathrm{P}_{1} &=   2 \left(1- F\left( \phi_1\right)\right)F\left(\tilde{\tau}\right)+\frac{1}{2}\sum^2_{i=1} \left(1 - \left[F\left(\phi_i\right)\right]^2\right) \\ \nonumber & - \sum^2_{i=1} F\left(\tilde{\tau}\right)\left(1 - F\left( \phi_i\right)\right) + \frac{1}{2} \left[\left[F\left( \tilde{\tau}\right)\right]^2 - \sum^2_{i=1}  \left[F\left(\rho\xi^*_{i}  \right)\right]^2 \right]^+\\ \nonumber &+ \frac{1}{2}\left[\left[F\left( \tilde{\tau}\right)\right]^2 -  \sum^2_{i=1}  \left[F\left(\rho\xi^*_{i}  \right)\right]^2 \right]^+,
\end{align}
for $\xi^*_{i}\geq0$, where   $[x]^+\triangleq max(0,x)$. Following   steps similar to the ones above, the outage probability at the strong user can be obtained as follows:
\begin{align}
\mathrm{P}_{2} &=
2\left[F\left(\tilde{\tau}\right) - F\left( \rho\xi^*_{2} \right)\right]^+\left[1-F\left(\tilde{\tau}\right)\right] \\ \nonumber &  +\sum^2_{i=1} F\left(\tilde{\tau}\right)\left[F\left(\tilde{\tau}\right)   - F\left( \rho\xi^*_{i} \right)\right]^+ -\\ \nonumber &\sum^2_{i=1}  \frac{1}{2} \left[\left(F\left(\tilde{\tau}\right)\right)^2 - \left(F\left( \rho\xi^*_{i} \right)\right)^2\right]^+  \hspace{-0.8em}+ \frac{1}{2}\sum^2_{i=1} \left(1-F\left(\phi_i \right)\right)^2.
\end{align}
\begin{lemma}\label{lemma1}
For the  case with two users in each group, the use of  one-bit feedback achieves  the same diversity order as that with perfect user ordering.
\end{lemma}
\begin{proof}Due to    space limitations, we consider only the case of $\tau<\min\left\{\frac{1}{\rho \xi^*_1},\frac{1}{\rho \xi^*_2}\right\}$.  The outage probability $\mathrm{P}_2$ can be approximated at high SNR as follows:
 \begin{align}
\mathrm{P}_{2} &\approx
2\left[\theta_2- \theta_0\right]\theta_0    +\sum^2_{i=1}\left[\theta_i   - \theta_0\right] +   \frac{1}{2}\sum^2_{i=1} \theta_i^2\\ \nonumber &-\sum^2_{i=1}  \frac{1}{2} \left[(1-\theta_0)^2 - (1-\theta_i)^2\right]\\ \nonumber &=
2\left[\theta_2- \theta_0\right]\theta_0     - \sum^2_{i=1}  \frac{1}{2} \left[\theta_0^2 - \theta_i^2\right]+   \frac{1}{2}\sum^2_{i=1} \theta_i^2,
\end{align}
where $\theta_0 =\frac{\tau^{N-\tilde{M}+1}}{(N-\tilde{M}+1)!( a_{k,q})^{N-\tilde{M}+1}}$ and $\theta_i=\frac{1}{(N-\tilde{M}+1)!(\rho a_{k,q}\xi^*_{i})^{N-\tilde{M}+1}}$. Since $\tau<\min\left\{\frac{1}{\rho \xi^*_1},\frac{1}{\rho \xi^*_2}\right\}$ and  both $\theta_i$ are at the order of $\frac{1}{\rho^{N-\tilde{M}+1}}$, the diversity order  is $2(N-\tilde{M}+1)$, the same as in Theorem \ref{theorem1}. The same conclusion can be made  for other choices of $\tau$ and    $\mathrm{P}_1$.
\end{proof}\vspace{-1em}

\section{Numerical Studies}
In this section, computer simulations are used to study the performance of the proposed massive-MIMO-NOMA scheme and verify the accuracy of the developed analytical results. This base station is equipped with a uniform circular array with $M=50$ isotropic antennas \cite{6542746,Dai15}. In Fig. \ref{fig fix2 1} the outage rate performance of the OMA and NOMA schemes is studied. As can be seen from the figure, the use of  the proposed massive-MIMO-NOMA scheme can yield a significant improvement in the system throughput. For example, at $\rho=20$dB and with $N=2$, the sum rate of the proposed NOMA scheme is three  times greater than  OMA. The accuracy of the developed analytical results for the case with perfect user ordering is verified in Fig. \ref{fig fix2 2}, and one can observe from the figure that the curves for the analytical results and simulations match perfectly.

In Fig. \ref{fig fix2 3} the impact of  one-bit feedback on the performance of the proposed MIMO-NOMA scheme is illustrated. As can be observed from the figure, the use of one-bit feedback achieves the same  diversity order as the case with perfect user ordering. For example, the slope of the curves for one-bit feedback is the same as that of the case with perfect feedback, which confirms the results shown in Lemma \ref{lemma1}. In addition, the user with strong CSI experiences less outage with one-bit feedback since it might be recognized as a user with poor CSI and is served with a smaller targeted data rate. \vspace{-1em}
\section{Conclusions}
In this letter, we have proposed a new massive-MIMO-NOMA scheme using limited feedback.  In particular, the use of the proposed scheme can decompose a massive-MIMO-NOMA system into  multiple SISO-NOMA channels, which significantly simplifies the design of MIMO-NOMA. Analytical results have been developed to evaluate the performance of the proposed scheme for   two scenarios, with perfect user ordering and with one-bit feedback, respectively.
\begin{figure}[!htp] \vspace{-0.2em}
\begin{center} \includegraphics[width=0.45\textwidth]{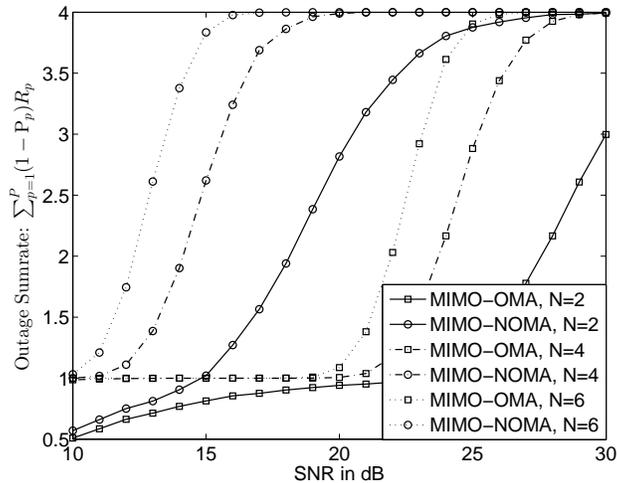}
\end{center}
\vspace*{-3mm} \caption{Outage sum-rates of two MIMO MA schemes. There are three users ($P=3$). $K=4$. $\alpha_1^2=\frac{5}{8}$, $\alpha_2^2=\frac{2}{8}$, $\alpha_2^2=\frac{1}{8}$,  $R_1=R_2=0.5$ bit per channel use (BPCU) and $R_3=3$ BPCU.  }\label{fig fix2 1} \vspace{-1.5em}
\end{figure}

\begin{figure}[!htp] \vspace{-0.2em}
\begin{center} \includegraphics[width=0.45\textwidth]{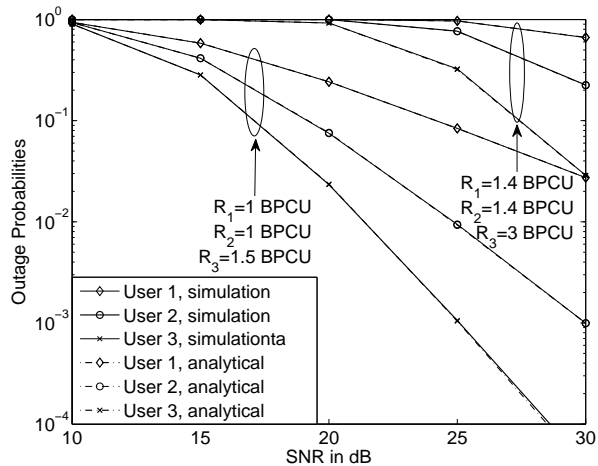}
\end{center}
\vspace*{-3mm} \caption{Accuracy of the developed analytical results with perfect user ordering.   There are three users ($P=3$). $N=2$ and $K=4$. $\alpha_1^2=\frac{5}{8}$, $\alpha_2^2=\frac{2}{8}$ and $\alpha_2^2=\frac{1}{8}$. $R_1=R_2=0.5$ bit per channel use (BPCU) and $R_3=3$ BPCU.  }\label{fig fix2 2} \vspace{-1em}
\end{figure}

\begin{figure}[!htp] \vspace{-0.2em}
\begin{center} \includegraphics[width=0.45\textwidth]{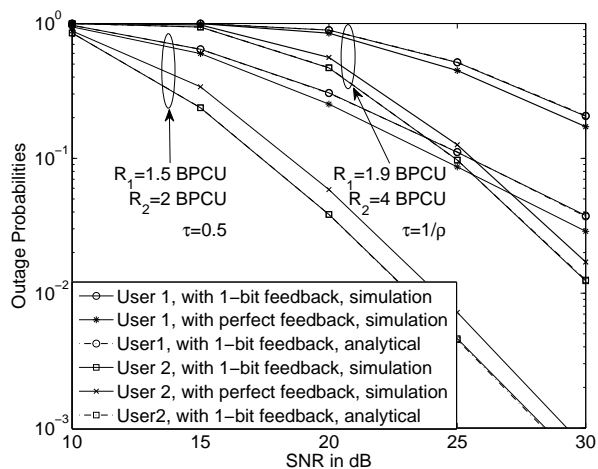}
\end{center}
\vspace*{-3mm} \caption{Outage performance of massive-MIMO-NOMA with different amounts of feedback. $P=2$,   $N=2$, $K=4$, $\alpha_1^2=\frac{3}{4}$, and $\alpha_2^2=\frac{1}{4}$.  }\label{fig fix2 3} \vspace{-2em}
\end{figure}

 \bibliographystyle{IEEEtran}
\bibliography{IEEEfull,trasfer}

  \end{document}